\documentclass[11pt]{article}
\usepackage[margin=1in]{geometry}
\usepackage{setspace}

\usepackage{fullpage}
\usepackage{amsmath}
\usepackage{amssymb}
\usepackage{amsthm}
\usepackage{epsfig}
\usepackage{pxfonts}
\usepackage{algorithm}
\usepackage{algorithmic}
\usepackage{color}
\usepackage{tikz}
\usepackage{enumerate}
\usepackage{appendix}
\usepackage{soul}
\usepackage{natbib}
\usepackage{url}
\usetikzlibrary{%
decorations.fractals,%
decorations.shapes,%
decorations.text,%
decorations.pathmorphing,%
decorations.pathreplacing,%
decorations.footprints,%
decorations.markings}

\newcommand{\tm}[1]{\textrm{#1}}
\renewcommand{\P}[1]{\mathbf{P}\left(#1\right)}

\newcommand{\E}[1]{\mathbf{E}\left[#1\right]}

\usepackage{soul}
\def\vm#1{{#1}}

\newtheorem{theorem}{Theorem}[section]
\newtheorem{lemma}[theorem]{Lemma}

\newtheorem{proposition}[theorem]{Proposition}
\newtheorem{corollary}[theorem]{Corollary}

\newtheorem{remark}[theorem]{Remark}
\newtheorem{claim}[theorem]{Claim}

\begin{document}

\title{A Generalized Bass Model for Product Growth in Networks}

\author{
Vahideh H. Manshadi%
    \thanks{ Yale School of Management, CT 06511.
    Email: \protect\url{vahideh.manshadi@yale.edu}.
    }
    \and
Sidhant Misra%
    \thanks{Los Alamos National Laboratory, Los Alamos, NM 87545.
    Email: \protect\url{sidhant@lanl.gov}.
    }
}

\maketitle

\begin{abstract}

Many products and innovations become well-known and widely adopted through the social interactions of individuals in a population. The Bass diffusion model has been widely used to model the temporal evolution of adoption in such social systems.  In the model, the likelihood of a new adoption is proportional to the number of previous adopters, implicitly assuming a global (or homogeneous) interaction among all individuals in the network.  Such global interactions do not exist in many large social networks, however. Instead, individuals typically interact with a small part of the larger population.  To quantify the growth rate (or equivalently the adoption timing) in networks with limited interactions, we study
a stochastic adoption process where the likelihood that each individual adopts is proportional to the number of adopters among the small group of persons he/she interacts with
(and not the entire population of adopters). When the underlying network of interactions is a random $k$-regular graph, we compute the sample path limit of the fraction of adopters. We show the limit coincides with the solution of a  differential equation which can viewed as a generalization of the Bass diffusion model.  When the degree $k$ is bounded, we show the adoption curve differs significantly from the one corresponds to the Bass diffusion model. In particular, the adoption grows more slowly than what the Bass model projects. In addition, the adoption curve is asymmetric, unlike that of the Bass diffusion model. Such asymmetry has important consequences for the estimation of market potential. Finally, we calculate the timing of early adoptions at finer scales, e.g., logarithmic in the population size.

\end{abstract}

\section{Introduction}
\label{sec:intro}

%



The diffusion of innovations and products via social interactions has long been observed in various social systems \cite{Roger}. In this process, a few pioneers adopt an innovation and influence those persons in contact with them. These people, in turn,   adopt  the innovation and  influence  their  contacts, and the  innovation spreads  through  the  network  as  a result  of these  social interactions. The impact of social interactions on the spread of innovation has intensified with the rapid growth and popularity of online social interactions.  Smartphone applications are a good example of a product whose popularity rests on the social interactions (and word-of-mouth communication) of the users.


Understanding the growth rate of the diffusion of a product/innovation in a population is imperative for both marketing and managerial decisions, such as inventory management and pricing.  Bass first studied the timing of product adoption using a simple differential equation later known as the {\em Bass diffusion model} \cite{Bass}.  In this  model, at  any  time,  the  growth  rate \vm{is proportional to}
the fraction  of the population who has  adopted it so far.   This implicitly relies on the assumption that the population is homogeneously
mixing, and thus, a non-adopter can be influenced  by {\em all} adopters. Such global interactions do not exist in many modern social networks, however.  Although the size of an online social network may be massive,  each person on that network is usually in contact with a small group of friends and he/she is influenced only by those persons.  This raises the following question:   How does product adoption evolve in networks with limited interactions? In this paper, we develop a machinery for characterizing product growth in a large population where the structure of the underlying network of interactions belongs to a large class of random graphs. \vm{When the underlying network of interactions is a $k$-regular random graph, we provide a generalization to the Bass model and contrast its adoption curve with the one corresponds to the Bass model. }

\vm{The Bass diffusion model is often presented in the following differential form:  $\frac{\partial s}{\partial t} \propto s\left(1 - s\right)$, where $s(t)$ is the
 fraction of adopters at time $t$.  \footnote{Note that this  is a special  case of the Bass model  where  the  coefficient of innovation is zero.}We show that when the underlying network is a random $k$-regular graph, the adoption process grows at the following rate} 

\begin{align}
\frac{\partial {s}}{\partial t} \propto \left[1- (1- {s})^{1-\frac{2}{k}}\right] (1-{s})
\label{eq:DE_Rand_intro}
\end{align}


\vm{The above differential equation implies, fixing the fraction of non-adopters ($1-s$), the growth rate depends sub-linearly on the fraction of adopters (as opposed to linearly in the Bass model).  This difference stems from the limited interactions in the network. However, as the degree $k$ grows, the factor $\left[1- (1- {s})^{1-\frac{2}{k}}\right]$ converges to  $s$.   Thus, the Bass diffusion model can be seen as a special case of differential equation \eqref{eq:DE_Rand_intro}.
}

\vm{Comparing the adoption curve resulted from these two differential equations, we show: (i) the adoption grows more slowly than what the Bass model projects, therefore, using the Bass model will over-estimate the adoption growth. (ii) unlike the adoption curve corresponding to the Bass model, the true adoption curve is asymmetric. In particular, the
adoption spreads faster in the second half (i.e., after reaching half the population) than the first half (see Figure \ref{fig:frac_vs_time} and Figure \ref{fig:CompvsRandom}). Such asymmetry can result in misestimating the future demand based on observations early in the adoption process.}

\vm{In order to establish the above result (and a few others), we study a natural stochastic adoption process which works as follows.}
Agents are nodes on a graph, and an edge between two nodes means these agents \vm{interact which each other}.  Initially, one random node adopts the product. Later, each node contacts a randomly selected neighbor at an independent Poisson process (with a given rate).   If the  contacting node is an adopter  and  the  contacted neighbor  is not,   the  latter adopts  the  product  with  a given probability.


Like \cite{Bass}, we are mainly concerned with the timing of the adoption  in the regime that a fraction of the population has already adopted. We denote this phase the {\em major adoption regime}, and we analyze the time it takes to increase the {\em fraction} of adopters by a constant (independent of the population size).   This is in the same spirit as fluid limits in queueing theory; \cite{Whitt_book,fluid_Book} \footnote{Note here we do not need to scale time, because the contact rate grows proportionally with the number of adopters.}. In particular, we show that the sample paths  of the scaled adoption  process (i.e., number  of adopters divided  by the  population size) almost  surely converges to a deterministic function.
First, as a sanity check, we find the limit for the complete network (\vm{which corresponds to the homogenously mixing population})
and confirm the limit of our adoption process coincides with the Bass model (see Theorem \ref{thm:complete} and Remark \ref{rem:comp}).  Next, and far more importantly, we establish the limit for random $k$-regular graphs (where $k$ is a constant independent of the population size) and \vm{show it equates with the solution of \eqref{eq:DE_Rand_intro}} (see Theorem \ref{thm:random} and Remark \ref{rem:random}).

We note that for similar stochastic diffusion processes, a differential equation approximation (particularly Kurtz’s theorem \cite{kurtz}) has been used to re-derive the Bass model when the underlying graph is assumed to be complete (e.g., \cite{Massoulie}). However, such an approach cannot be directly  applied to other  network  structures. In fact, analyzing  the scaled sample  path  of the  diffusion process  for general  graphs  is prohibitively difficult. Recent  work has used concepts  from mean field theory  to approximate the growth  rate  for a certain  class of random  graphs \cite{Jackson_Rogers,Ramesh_Shakkottai,Yaniv}.  The basic idea of these  models is to approximate the  fraction  of adopter  neighbors  of each agent by the  fraction of adopters in the  whole population.  \footnote{When  dealing  with  random graphs  with  general  degree  distribution, some  of these  methods approximate the fraction of adopter neighbors with a given degree $d$ by the fraction of adopters with degree $d$ in the whole population.}
\vm{However, our analysis shows that when the degree is bounded, an adopter is more likely to be connected to adopters. Therefore, the fraction of his adopter neighbors will be higher than the fraction of adopters in the population. }
\vm{Accordingly, the aforementioned method of approximation tends to over-estimate the  growth  rate  in random  regular  graphs,  even though  it does improve  upon the  Bass model.  }
In order to exactly characterize the growth rate,  we develop a new technique  to incorporate the effect of network structure in the evolution of the adoption  process for random  regular graphs.  In Section \ref{sec:discus}, we describe how to generalize our analysis to compute  the limits for random  graphs with more general  degree distributions (under  certain  conditions  for the  distribution).  We also explain  how to modify our analysis to accommodate a more general Bass model with innovators  and an SIR epidemic model.  \footnote{These  terminologies are used  in  infectious  diseases;  S stands for susceptible, I for infectious  and  R for removed, i.e.  immune  or dead.}


In addition  to finding the limits in the major adoption  regime, we find the limit of timing  in the  early adoption  regime which refers to the  phase when the number  of adopters is logarithmic  in the  population size.  In this  regime, we show that the  time  needed  to acquire a logarithmic  number  of adopters scales double logarithmically  in the population size. We also show that compared  to complete graphs, the  adoption  process  grows  more  slowly  in  random  $k$-regular  graphs  by  a  factor  of  $\frac{k}{k-2}$ (see Theorem  \ref{thm:sublin}).

\vm{Besides marketing decisions, several managerial  questions  arise in the  presence  of product  diffusion (sometimes  called social learning  or word-of-mouth  effects), for  example,  how to manage  demand  when the  supply is constrained \cite{Savin,kumar,Shen_Supply}, how to price optimally   \cite{Shen_pricing}, 
\footnote{For pricing in presence of externalities or word-of-mouth effect in social networks, see also \cite{Campbell,Khakbod,Ifrach,Ozan, maxim}.} and how to facilitate  diffusion for the adoption  of green technologies  \cite{Saed}\cite{Diaz}. The Bass model has been extensively  used to study  these  problems, but in situations where each agent  can only influence a limited  number  of others,  it tends  to misestimate the  adoption growth.   Accordingly, our proposed technique  for exactly characterizing the  product  growth  can prove useful in developing \vm{more accurate} managerial  insights  for product diffusion in networks with limited  interactions.}

\subsection{Related Work}
\label{subsec:related}


The  adoption  model  we study  in this  paper  has  counterparts in many  other  areas,  ranging from epidemiology to economics.  It  is a stochastic  version of the  SI (Susceptible  Infected)  model used  in epidemiology \cite{epidemic,Durrett}. It is also closely linked to the  gossip algorithms  used in data  aggregation  and  distributed computing  in sensor and  peer-to-peer  networks  \cite{Devavrat_book,Devavrat_gopssip}, and to the learning  models studied in social and  economic  networks  \cite{Jackson}.  Our  work complements the  earlier  studies  of epidemic processes by focusing on  timing  in the  major  adoption  regime (i.e., when a fraction of the  population has  already  adopted).  Several  papers  have shown the growth of epidemics  is related  to  the  spectral  radius  and  expansion  properties  of the  underlying  graphs \cite{Devavrat_gopssip,Jackson} (for more general epidemic models, see \cite{GaneshSIR,GaneshSIS,SanjaySIS}).  It is well known that because random graphs have large expansion factors \cite{expansion1,expansion2}, epidemics spread  fast on them.   However,  the  functional  form of the  growth  has not  previously been calculated, nor has  the time  needed  to grow the  fraction  of nodes in epidemics  from a constant  $\alpha$ to another  constant  $\gamma$ been analyzed.


As mentioned above, several papers use approximation methods to develop tractable frameworks to analyze the  adoption/epidemic process in networks  with  a given degree distribution.  \cite{Jackson_Rogers} use a mean field approach to study an SIS (Susceptible  Infected Susceptible)  model and relate stochastic  dominance  properties  of the  degree distribution to  the  infection  rate.   \cite{Ramesh_Shakkottai} employ a mean field approximation to model the  temporal  evolution  of demand  for stored  content on the  Internet.  Using this  demand  model,  they  examine  the  delay  performance  of several content distribution mechanisms. \vm{Even though such a mean field approximation method provides a tractable framework, our analysis show that it tends to over-estimate the demand growth by neglecting the phenomenon that an adopter is more likely to be connected to adopters and such correlation evolves over time. }


Several papers  in the area of marketing are concerned with adoption  processes and the flow of information  in networks (see \cite{Yaniv} and references therein). Closest to our work are \cite{Yaniv}  who study  an adoption  process similar to ours on a random  graph  with a given degree distribution.  They  use an approximation method  similar  to that of \cite{Ramesh_Shakkottai} to show the adoption  growth  rate  depends  on the  mean and  variance  of the  degree distribution. They  use this model to uncover the degree  distribution based on the observed adoption  data.  Further, they  show when  the  network  degree is significantly  skewed,  the  adoption  curve  is asymmetrical. This has important consequences  for the estimation of market  potential. Interestingly, our rigorous analysis also shows the adoption  curve is asymmetric, even on random  $k$-regular graphs.

\vm{From a technical perspective, our} work  brings  the  literature on    processes  on  random  graphs  together with that on stochastic   differential  equations   and  fluid  limits.    To  analyze  the  adoption   process  on  random graphs,  we couple the  (continuous) adoption  process with  the  (discrete) graph  generation  process based on a configuration  model \cite{Wormald}. Abstracting from time, our adoption  process spreads  on random  graphs  in the same way as the {\em exploration  process} defined in \cite{molloy1}.   The  latter process was introduced to  find the  size of the  largest  connected  component in a random  graph  with  given  degree  distributions.   We  use    ideas  similar to  \cite{wormaldDiff2} to approximate the evolution of the adoption  process.  When coupling with time, we build on these results  \cite{wormaldDiff2,molloy1} for random  graphs  to compute  the limit of timing of the adoption  process.

%
%

%
%
%
%

\section{Model and Main Results}
\label{sec:modelNmain}

We represent the social network  by graph  $G_n = (V,E)$, where $|V| = n$.  Each  node $v \in V$ represents an agent in the  system;  nodes $v$ and $u$ are neighbors  if $(v,u) \in E$.  At time  $0$, a randomly  selected node adopts  a new product  $Z$.  The  new product  spreads  through  the  local contacts between  the neighbors.  In particular, each node $v \in V$, contacts  a randomly  selected neighbor at an independent Poisson  process with  rate  $\beta$.  Suppose node $v$ adopts  $Z$ at  time  $t$; at  any contact after  $t$, if node $v$ contacts  a neighbor  $u$ that has not yet adopted  $Z$,  $u$ will adopt  the product  with probability $p$. Given the thinning  property of the Poisson process, WLOG,  we assume $p = 1$.


In this adoption  process, the number  of adopters can only increase over time.  If the underlying graph  is connected,  after  a finite time,  all agents  will adopt  the  new product. For any $1 \leq x  \leq n$, let $T_n(x)$ denote  the  minimum  time  needed  to have $x$ adopted  individuals.   Our  goal is to analyze limits  of $T_n(x)$ for different scales of $x$.  In particular, we define two main  regimes:  an {\em early adoption} regime in which $x  = O(\log n)$ and a {\em major adoption} regime where $x = \Theta(n)$.


We analyze  the  adoption  process on two classes of graphs:  complete  graphs  and  random  $k$-regular graphs where $k$ is a constant. The former class represents a network with global \vm{(or homogeneous)} interactions, and the latter serves as a model of limited  interactions among individuals.


In the next section,  we focus on the major adoption  regime and give almost  sure results  on how long it takes to grow the fraction  of adopters from $\alpha$ to $\gamma$, where $0 < \alpha \leq \gamma < 1$.

\subsection{Timing in Major Adoption Regime}
\label{subsec:major}


In the major  adoption  regime, we assume a constant fraction  of the population has already adopted  the product, and we are concerned with the time needed to add $\Theta(n)$ more adopters.  More precisely, for any $0  < \alpha \leq \gamma < 1$, let $\Delta_n(\alpha n,\gamma n )$ be $T_n(\gamma n) - T_n(\alpha n)$.  In this subsection,  we compute the limit of $\Delta_n(\alpha n,\gamma n )$. We start by analyzing  the timing  in major  adoption  when the underlying graph  is a complete  graph  and show that:

\begin{theorem}[Major  adoption  in a complete  graph]
Suppose for all $n >1$, the underlying  graph $G_n$ is the complete graph.  Then,  for any $0  < \alpha \leq \gamma < 1$:

\begin{align}
\Delta_n(\alpha n,\gamma n ) \stackrel{a.s.}{\rightarrow}  \theta(\gamma) - \theta(\alpha),
\label{eq:timeLimit}
\end{align}
where $\theta(s) = \frac{1}{\beta}\log \frac{s}{1-s}$,  for $0 < s < 1$.
\label{thm:complete}
\end{theorem}


First  note that function  $\theta(s)$ is centered  such that $\theta(1/2) = 0$.  Also, note that function  $\theta(s)$ is strictly  increasing,  and, thus, it is an injective function.  Its inverse is $s(t) = \frac{e^{\beta t}}{1 + e^{\beta t}}$, {\em logistic equation}  that is a special case of the Bass model.


\begin{remark}
\label{rem:comp}
Let $S_n(t)$ be the number of adopters  at time $t$.  Theorem  \ref{thm:complete} implies that  $\frac{S_n(t)}{n} \stackrel{a.s.}{\rightarrow}  s(t)$, where $s(t)$ is the solution  of the following differential  equation:

\begin{align}
\frac{\partial s}{\partial t} = \beta s (1-s)
\label{eq:DE_Comp}
\end{align}

\end{remark}




Thus,  the  limit of the  scaled sample  paths  of our probabilistic adoption  process coincides with the  deterministic logistic  function.
Further, note  that function  $\theta(s)$ has  the  following symmetry property:  for any  $0 < s < 1$, $\theta(s) = - \theta(1-s)$.   The  time  it  takes  to  grow the  fraction  of the  adopters from $\alpha$ to $1/2$ is the  same as the time  it takes  to grow the  fraction  from $1/2$ to $1-\alpha$ where $0 < \alpha < 1/2$.  This  symmetry  results from having a complete  graph  (a \vm{homogenously} mixing population) and is intuitively explained as follows: at any time,  the subgraph  including the adopters is a complete  graph,  and so is the subgraph  consisting of nodes who have not yet adopted. Now, we can look at the process in a backward  way; if node $u$ is a non-adopter, and it contacts  node $v$ who is an adopter, then  node $v$ will abandon  product $Z$; the abandonment of the  product  will spread  through the network  in this way.  Because the processes of adoption  and discard  spread  in exactly  the same way, the time needed to grow the set of non-adopters from $1 - s$ to $1/2$ will be the same as the time required  to grow the set of adopters from $s$ to $1/2$.


The  limit  result   \eqref{eq:DE_Comp} can be proven  directly  by using stochastic  differential  equations  and  Kurtz’s  theorem  (for  instance,  see \cite{Massoulie}, Section  1.3.1).   Here,  we present an alternative proof that analyzes  the  random  times  between  any two consecutive  adoptions  and directly  establishes  the time limit \eqref{eq:timeLimit}.

\begin{proof}

For any $1 \leq i \leq n-1$, let $\tau_i$ be the time it takes  to grow the number  of adoptions  from $i$ to $i + 1$.  First  note that $\Delta_n(\gamma n ,  \alpha n) = \sum_{i = \alpha n }^{\gamma n -1}\tau_i$. Further, note that conditioned  on the  set of adopted
nodes, $\tau_i$'s are independent exponential  random  variables.  Let $\lambda_i$ be the rate of $\tau_i$. For the complete
graph,  we compute  the rate  $\lambda_i$ as follows: there  are $i$ adopter  nodes who can contact non-adopters. When  any adopter  node $v$ makes a contact, it contacts  a neighbor who has not adopted  yet with probability $(n-i)/(n-1)$.  Thus,  using the thinning  property of the Poisson processes, we have:
\begin{align}
\lambda_i = \beta  \frac{i(n-i)}{n-1}.
\label{eq:rate:complete}
\end{align}

\noindent{First, to prove the theorem, we make the following claim:}

\begin{claim}
$\E{\Delta_n(\gamma n, \alpha n)} \rightarrow  \theta(\gamma) - \theta(\alpha)$.
\label{claim:limit}
\end{claim}

\noindent{The claim is proven in Appendix  \ref{appendixA}. Next, we establish  the concentration bounds shown in the second claim, given as:}

\begin{claim}
Suppose $\epsilon$ is a fixed small positive number,
\begin{align}
\P{|{\Delta_n(\gamma n, \alpha n)} - \E{\Delta_n(\gamma n, \alpha n)}| \geq \epsilon} \leq e^{-\delta n},
\label{eq:claim:hoeffding}
\end{align}
where $\delta$ is a small positive number given in Equation \eqref{delta}.
\label{claim:concentration}
\end{claim}

\noindent{The second claim is proven in Appendix  \ref{appendixA} as well. This implies that}
\begin{align*}
\sum_{n = 1}^{\infty} \P{|{\Delta_n(\gamma n, \alpha n)} - \E{\Delta_n(\gamma n, \alpha n)}| \geq \epsilon} < \infty .
\end{align*}

\noindent{Now, applying  the  Borel-Cantelli lemma,  we have: $\Delta_n(\alpha n,\gamma n ) \stackrel{a.s.}{\rightarrow} \E{\Delta_n(\alpha n,\gamma n )}$ which  completes  the
proof.}

\end{proof}


Next,  we analyze  the  adoption  process on random  $k$-regular graphs,  where $k$ is a constant.  To ensure all nodes eventually  adopt,  we limit the sample space of the graphs to only include connected  ones.

\begin{theorem}[Major adoption in random $k$-regular graphs]
Suppose for all $n>1$, the underlying graph  $G_n$ is sampled  uniformly  at  random  from  the  set  of all connected  $k$-regular  graphs  with n nodes, where $k \geq 3$ is bounded.  For  any $0  < \alpha \leq \gamma < 1$, the following limit holds:
\begin{align}
\Delta_n(\gamma n,\alpha n) \stackrel{a.s.}{\rightarrow}  \tilde{\theta}(\gamma) - \tilde{\theta}(\alpha),
\label{eq:timeLimitR}
\end{align}
where $\tilde{\theta}(s) = \frac{k}{\beta(k-2)}\left[  \log{\left(1 -  (1-s)^{\frac{2}{k}-1} \right)}  - \log \left( 1 - 2^{1 -\frac{2}{k}} \right) \right]$.
\label{thm:random}
\end{theorem}

\begin{remark}
\label{rem:random}
Let $\tilde{S}_n(t)$ be the number of adopters  at time $t$.  Theorem  \ref{thm:random} implies $\frac{\tilde{S}_n(t)}{n} \stackrel{a.s.}{\rightarrow}  \tilde{s}(t)$, where $\tilde{s}(t)$ is the solution  of the following differential  equation:
\begin{align}
\frac{\partial \tilde{s}}{\partial t} = \beta \left[1- (1- \tilde{s})^{1-\frac{2}{k}}\right] (1-\tilde{s})
\label{eq:DE_Rand}
\end{align}

\end{remark}


Figure  \ref{fig:frac_vs_time} compares the solution  of differential  equations  \eqref{eq:DE_Comp} and \eqref{eq:DE_Rand} for initial  value $s(0) = 0.01$, and $k = 5$.  As we can see, the  adoption  grows much  more slowly on a random  $5$-regular  graph  than on a  complete  graph.   The same can be observed  in the  left plot of Figure  \ref{fig:CompvsRandom} which basically  shows the  inverse  function.   (Note: the left  plot  of Figure  \ref{fig:CompvsRandom} shows the  limit  results  of the  time  it  takes  to  grow the  fraction  of adopters from $0.01$ to $s \in [0.01,0.99]$ on the  complete  graph  and the random  $5$-regular graph.)   The  high level intuition behind  this  observation is as follows: suppose  we reach  the  time that $i$ nodes have already adopted, where $i = \Theta(n)$; the rate of contact of adopters is $\beta i$ regardless of the underlying  graph.  However, the probability that an adopter  contacts  a non-adopter is higher on a complete  graph  for two reasons.  First, on a $k$-regular  graph,  the  subgraph  induced  by the adopters is connected; therefore,  each  adopter  has  $k - 1$ neighbors  who are  likely to be non-adopters.  Second,  those  $k - 1$ neighbors  are  not  uniform  samples  among  the  remaining $n - 2$ nodes.  In fact,  we show it is more likely that the  neighbor  of an adopter  belongs to the  set of adopters  itself.  This is a result  of the connectivity properties  of the subgraph  induced  by the adopters.

%

To further  highlight the  effect of connectivity among the  adopters, let us compute  the  rate  of $\tau_i$ on a random  graph  using a {\em mean  field approximation} (in the  same spirit  of approximation as \cite{Jackson_Rogers,Ramesh_Shakkottai,Yaniv}). We denote  this approximate rate  as $\tilde{\lambda}^{M}_i$.  As explained  above, the  rate  of contact by adopters is $\beta i$; each adopter has at  least  one adopter  neighbor  with  probability $1$.  For  adopter  node $v$, let $v'$  be the  neighbor who adopted  before $v$ and  was the  first adopter  who contacted $v$.  Clearly,  if $v$ contacts  node $v'$, this will not result  in a new adoption.  Now suppose node $v$ selects a random  neighbor  other  than $v'$; this happens  with probability $(k - 1)/k$.  In a mean field approximation, we assume the rest of the neighbors  of $v$ are uniformly  sampled  among the  other  $n-2$ nodes.  Thus,  the probability node $v$ will contact a non-adopter is $\frac{n-i}{n-2}$. This implies

\begin{align}
\tilde{\lambda}^{M}_i =\frac{\beta (k-1)}{k}  \frac{i(n-i)}{n-2}.
\label{eq:rate:meanF}
\end{align}


The approximation rate  of \eqref{eq:rate:meanF} has the same form as \eqref{eq:rate:complete} and is only scaled by $(k-1)/k$. The time limit resulting  from this mean field approximation is plotted in Figure  \ref{fig:CompvsRandom}. As we can see, it significantly  differs from the  actual  limit,  and,  in particular, it underestimates the  adoption time.

\begin{figure}[tbh]
\centering
{\includegraphics[width=17cm,height=7.5cm]{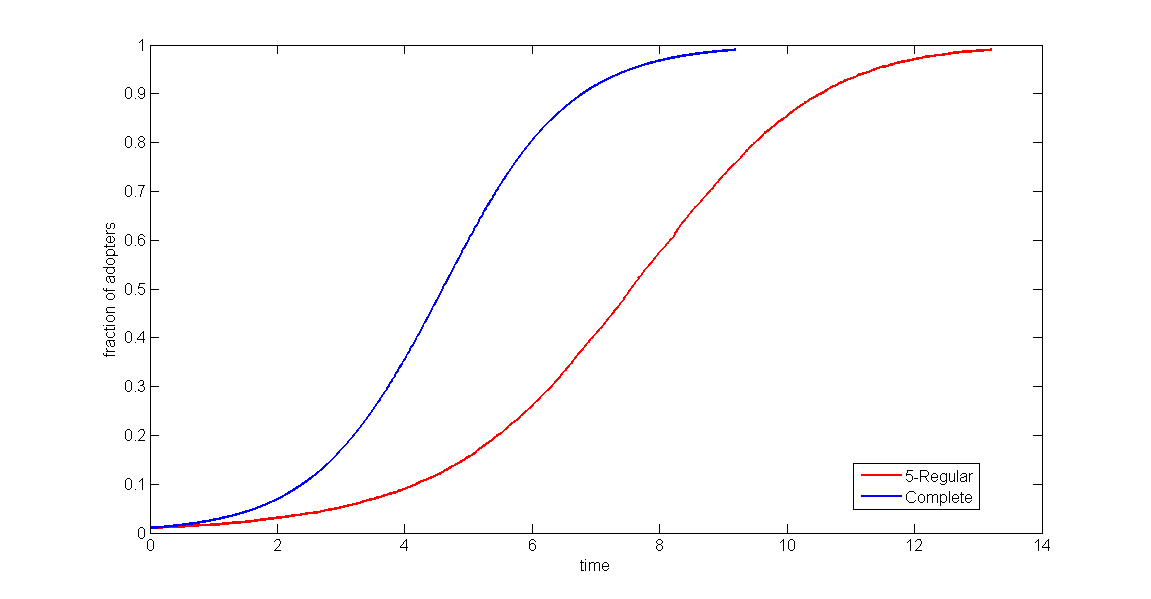}}
\caption{{\protect {Evolution of the fraction of adopters ($s(t)$); comparing a complete and a 5-regular random graph. Time is normalized such that at time 0 the fraction of adopters is 0.01.}}}
\label{fig:frac_vs_time}
\end{figure}


Further, on the right plot of Figure  \ref{fig:CompvsRandom}, we observe that unlike the complete graph, the normalized process for the random  regular  graph  (such that $\tilde{\theta}(1/2) = 0$) is not symmetric  around $1/2$, and the time  it  takes  to grow the  process from $\alpha < 1/2$ to $1/2$  is larger  than  the  time  it  take  to grow it
from $1/2$ to $1 - \alpha$.  This is again related  to the connectivity properties  of random  graphs.  If we look at the backward  process, we see the subgraph  of non-adopters is not necessarily connected;  thus,  the backward  process grows faster.


\begin{figure}[tbh]
\centering
{\includegraphics[width=17cm,height=7.5cm]{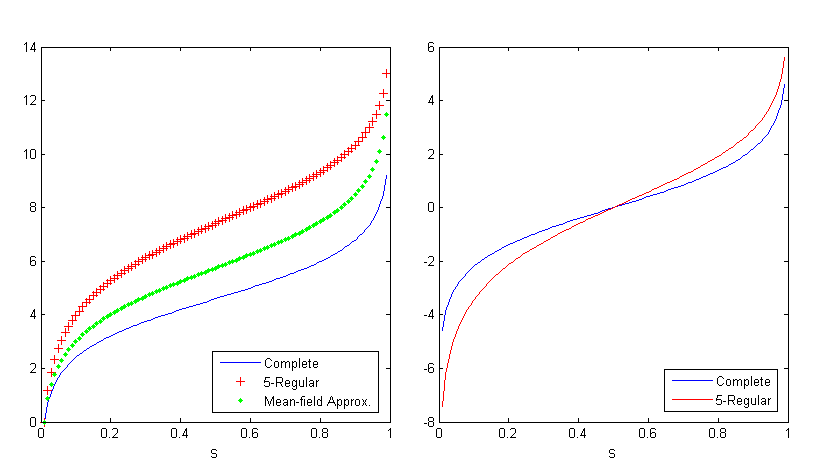}}
\caption{{\protect  {Left: The time it takes to grow the fraction of adopters from $0.01$ to $s \in [0.01,0.99]$ (functions $\theta(s) - \theta(0.01)$, $\tilde{\theta}(s) - \tilde{\theta}(0.01)$, and  $k/(k-1)\left(\theta(s) - \theta(0.01)\right)$ for $s \in [0.01,0.99]$). Right: The normalized timing functions  in major adoption regime (functions $\theta(s)$ and $\tilde{\theta}(s)$ for $s \in [0.01,0.99]$).}}}
\label{fig:CompvsRandom}
\end{figure}

{\bf Proof sketch of Theorem \ref{thm:random}:}

Similar  to  the  proof  for the  complete  graph,  we compute  the  rate  of the  exponential  time  between  any two  consecutive  adoptions.   Recall that we denote  the time it takes  to grow the number  of adoptions  from $i$ to $i + 1$ as $\tau_i$ and its rate by $\lambda_i$. Unlike the complete  graph,  we cannot  compute  the rate  $\lambda_i$ only based on the number  of adopters: suppose  node $v$ is an adopter,  and  it  samples  one of its  $k$ neighbors  to  contact.  Knowing  only $i$, we cannot  determine  how many  of $v$'s neighbors  have not  yet adopted.  To overcome this problem,  we first note that the  random  graph  can be generated  using an iterative pairing  process called the  {\it configuration  model} \cite{Wormald}. A configuration model works as follows: we start with $n$ isolated nodes.  Each node has $k$ clones (or half edges). At each step, a new edge is formed by pairing  two randomly  chosen clones; the process ends after  $nk/2$ steps,  when all the clones are paired.


Given  this  observation,  we couple  the  graph  generation and  the adoption  process in the following manner.:  We assume  the  graph  has  not  been  realized  before the  adoption  process. Thus,  at time $0$, we have $n$ isolated  nodes, each with $k$ unpaired  clones. Any time  an adopter  makes a contact, it chooses one of its clones uniformly at random.  If the clone has already  been paired,  this means both  ends of this edge have already  adopted. In this case, neither  the adoption  process nor the  set of formed edges will grow.  If the  clone has not  been paired,  we perform  a  {\it new iteration} of the  configuration  model and  form a new edge by sampling  a clone at random  among all unpaired  ones.  If the sampled  clone belongs to a non-adopter node, then  the adoption  process grows by one; otherwise,  the adoption  set remains  the same.


Abstracting away from the time, we first analyze the evolution of the number  of adopted  nodes and  their unpaired  clones in terms  of the  number  of iterations (which  are  discrete  time  random processes).  For the analysis,  we use the differential  equation  method  proposed by \cite{wormaldDiff,wormaldDiff2} to approximate discrete  random  processes using a deterministic function.   The  same approach has been used to find the  size of the  giant connected  components  in random  graphs by \cite{molloy1,molloy2}.


Next, we compute  the exponential  times between  any two consecutive  iterations, and based on the result, we compute  $\E{\Delta_n(\gamma n,\alpha n)}$.  At the  end,  similar  to  the  proof of Theorem \ref{thm:complete}, we use some concentration bound  and the Borel-Cantelli lemma to establish  almost  sure convergence.


The  proof also deals with  some  technical  subtleties; for example, it obtains  a bound  in a deterministic approximation stronger  than  the bound established in \cite{wormaldDiff,wormaldDiff2}. It also confirms the  concentration results  hold when we limit  the  sample  space to simple and connected  graphs  (rather than  all possible pairings  of the clones.).



The detailed  proof of Theorem  \ref{thm:random} is presented  in Section \ref{sec:random}. In Section \ref{sec:discus}, we explain how to  generalize  the  above  proof ideas  to  analyze  the  adoption  process  on a random  graph  with more general  degree distributions.  We also describe  how to modify our analysis  to more general Bass model with innovators and an SIR (Susceptible-Infected-Remove) epidemic model.



Finally, note that for $k = 2$, the only connected regular graph is a cycle of length $n$. The following proposition asserts  the different time scalings of the major adoption  regimes for such graphs:
\begin{proposition}
\label{cor:cycle}
For  all $n>1$, let the underlying  graph  be a cycle of length $n$.   Then,  for any $0  < \alpha \leq \gamma < 1$,

\begin{align*}
\frac{\Delta_n(\gamma n,\alpha n)}{n} \stackrel{a.s.}{\rightarrow}  \frac{\gamma - \alpha}{\beta}.
\end{align*}

\end{proposition}

\begin{proof}

For  any  $i \geq1$, the  set  of adopters forms a path  of length  $i$.  For  any  $1 <  i <  n$, there  are only two nodes at  the  two ends of the  path  that can contact non-adopters. Suppose node $v$ is an end point, and it makes a contact. With  probability $1/2$, it contacts  a non-adopter neighbor.  Thus, $\lambda_i = \beta$ for any $1 < i <n$, implying that $\tau_i$'s are i.i.d.  By the strong  law of large numbers,  the above limit holds.

\end{proof}

\subsection{Timing in Early Adoption Regime}
\label{subsec:early}


In a major adoption  regime, the rate  of contacts  grows linearly with $n$; \footnote{Except for the special case of a single cycle. }   therefore,  the adoption process spreads  very quickly. Further, as Theorem  \ref{thm:complete} and \ref{thm:random} assert,  the  time  to grow from fraction  $\alpha$ to fraction $\gamma$ is constant.  However, in the  early adoption  regime, the  growth  rate  is much  slower, as there  are only a few adopters. In this  section,  we analyze  the  timing  in this  slower regime; more specifically, we find the limit of the time it takes  to get $\Theta(\log{n})$  adopters. We establish  the limit for both the complete graphs and the random  $k$-regular graphs.  For both graphs, we establish that the time needed to have $\Theta(\log{n})$ adopter scales as $\log{\log{n}}$.  Further, we show the process grows faster  on the complete  graph  (compared  to the random  $k$-regular graph).

\begin{theorem}[Early adoption] For any constant $C>0$ the following hold:
\begin{description}
\item[(a)] If for all $n>1$, graph $G_n$ is a complete graph, then:
\begin{align}
\frac{T_n(C \log{n})}{\log{\log{n}}}  \stackrel{p}{\rightarrow}  \frac{1}{\beta}.
\end{align}
\item[(b)] If for all $n>1$, graph $G_n$ is a uniformly random sample from the set of all connected $k$-regular graphs, where $k \geq 3$ is bounded, then:
\begin{align}
\frac{T_n(C \log{n})}{\log{\log{n}}}  \stackrel{p}{\rightarrow}  \frac{k }{\beta (k-2)}.
\end{align}
\end{description}
\label{thm:sublin}
\end{theorem}


The proof of part (a) follows the same line of the proof of Theorem \ref{thm:complete}; note that the rate given by \eqref{eq:rate:complete} holds for any $1 \leq i < n$.  Using this rate,  we first show that $\frac{\E{T_n(C \log{n})}}{\log{\log{n}}} \rightarrow \frac{1}{\beta}$, and then show $\frac{T_n(C \log{n})}{\log{\log{n}}}$  converges to its mean (in probability) by proving its variance  converges to zero. The details  are given in Appendix  \ref{sec:sublin}.

{\bf Proof sketch of part (b)}

As in the  previous  proofs, we aim to compute   $\lambda_i$, for $1 \leq i \leq C \log{n}$.  First  note that we can view the adoption  process as the following:  each edge $(v,u)$ makes contacts  at an independent Poisson rate  $\beta/k$ from $v$ to $u$  and similarly from $u$ to $v$. Given $i$ adopters, the total  number  of contacts  (along edges) that can result  in a new adoption  is the total number  of edges between  the  set of adopters  and  non-adopters.  Observe  that in our coupled process (defined  in proof sketch  of Theorem  \ref{thm:random}), at  each iteration, the  sub-graph  of the  adopters formed  so far is connected;  thus,  the  total  number  of edges that can  result  in a new adoption  is, at  most,  $(k-2)i +2$.   To  compute  the  exact  number  of these  edges,  we use the  {\it locally tree-like} property of random  $k$-regular graphs   \cite{dembo} which confirms that, with high probability, on the  final realized  graph  (i.e.,  after  the  formation  of all edges),  the  subgraph containing  these $i$ nodes is a tree and does not contain  a cycle. Thus, the remaining   $(k-2)i +2$ edges are all between  an adopter  node and  a non-adopter one.  This  implies that with  high probability, $\lambda_i \approx \beta/k [(k-2)i]$ for $1 \leq i \leq C \log{n}$.   Once we have  the rates,  the  rest  of the  proof  handles  the convergence of random  variable $T_n(C \log{n})$ to its mean.  The  detailed  proof is given in Appendix \ref{sec:sublin}.


The locally tree-like property of the random  graph  carries over to a subgraph of size $o(\sqrt{n})$; thus, we can use similar techniques  to prove the following:

\begin{remark}
Similar limit results hold for $\frac{T_n(\sigma(n))}{\log{\sigma(n)}}$ for any $\sigma(n) = o(\sqrt{n})$.
\end{remark}

\section{Proof of Theorem \ref{thm:random}}
\label{sec:random}
In this section we formally prove Theorem \ref{thm:random}. For the sake of brevity, the proof of all technical lemmas of this section is deferred to Appendix \ref{appendixB}.
As explained in the proof sketch, we first study the (edge formation) iteration process and analyze the evolution of the number of adopters in terms of the iterations. Similar to \cite{molloy1,molloy2}, we call this the {\it exploration process}. In this process every node is associated with $k$ clones. For clone $c$ of node $v$, all other clones belonging to node $v$ are considered as $c$'s siblings. We start at iteration $j = 0$. At any iteration $j$ in the exploration process, there are three kinds of clones: `sleeping' clones, `active' clones, and `dead' clones. At the beginning, all clones are sleeping. If all clones of a node are sleeping the node is said to be a sleeping node; if all its clones are dead, the node is considered dead; otherwise, it is considered active. Given this terminology, the exploration process works as follows.

\vspace{0.2in}
\hspace{-0.1in} \textbf{Exploration Process}
	\begin{itemize}
		\item[1.] \emph{Initialization}: Pick a node uniformly at random from the set of all sleeping nodes and set the status of all its clones to active.
	  \item[2.] Repeat the following two steps as long as there are active clones:
	\begin{itemize}
		\item[(a).] Sample a clone $c$ uniformly at random from the set of active clones and kill it.
		\item[(b).] Pair the clone $c$ with clone $c'$ that is chosen uniformly at random from the set of all remaining unpaired (active or sleeping) clones. Kill $c'$ and make all its siblings active.
	\end{itemize}	
	\end{itemize}
	\vspace{0.05in}

The set of active clones at iteration $j$ is denoted by $A(j)$. The union of the set of sleeping and active clones is denoted by $L(j)$; they are called `living' clones.
Finally, we denote the number of sleeping nodes by $N(j)$.

\subsection{Evolution of Exploration Process}
\label{subsec:evolution}

In a second paper, \cite{molloy2} used results by \cite{wormaldDiff2} to track the evolution of the exploration process. We employ a similar technique, but do not directly use Wormald's result. Instead, we use insights from the proof technique and tighten the probability of error needed to get almost sure convergence. When coupling with the adoption process, we introduce an additional variable which tracks the (random) re-scaled time in the {adoption} process.

 At every iteration of the exploration process, the number of living clones reduces by two; i.e., we have $L(j+1) = L(j) - 2$. Hence we have $L(j) = nk - 2j$. At each iteration, the number of sleeping nodes, $N(j)$, reduces by one, if in step $2(b)$ of the exploration process, the clone neighbor chosen ($c'$) is a sleeping one. Otherwise, $N(j)$ remains the same.

After initialization, all nodes are sleeping except the one we have awakened for the initiation (note that when coupling with the adoption process, this will be the first node to  adopt the product). Therefore, at $j=0$, we have $A(0) = k$.

	At every iteration, the evolution of the number of sleeping nodes and the number of active clones is as follows:
	\begin{enumerate}
		\item With probability $\frac{k N(j)}{kn - 2j }$, we have
		\begin{align} \label{evolutionstart}
			N(j+1) &= N(j) - 1 \\
			A(j+1) &= A(j) + (k-2).
		\end{align}
		\item With probability $\left(1 - \frac{k N(j)}{kn - 2j }\right)$, we have
		\begin{align}
			N(j+1) &= N(j) \\
			A(j+1) &= A(j) -2. \label{evolutionend}
		\end{align}
	\end{enumerate}

Note that the above equations (for the evolution of sleeping nodes and active clones) only hold when the graph is connected. Later, in Lemma \ref{lem:connectivity}, we show a random $k$-regular graph is connected with probability $1 - \Theta(n^{-2})$. We further show that our limit result holds when conditioning on being connected. For now, we assume the graph is connected and find the limits of the scaled random variables, $N(j)/n$ and $A(j)/n$, evolving according to \eqref{evolutionstart}-\eqref{evolutionend}.

	 At any iteration $j$ in the exploration process, we have $A(j) = L(j) - kN(j) = k[n - N(j)] - 2j$. So it suffices to characterize the evolution of only one of these parameters, e.g., $N(j)$.
Let $H(j)$ denote the history of the exploration process until iteration $j$. By the equations (\ref{evolutionstart}) - (\ref{evolutionend}), we have
\begin{align} \label{eq:sleeping}
\E{N(j+1) - N(j) | H(j)} = -\frac{k N(j)}{kn - 2j} = -\frac{k N(j)/n}{k - jt/n}.
\end{align}

From the above, using Wormald's result (Theorem 1 in \cite{wormaldDiff2}), it follows that for any $0 < j < (k - \epsilon_0) n/2$, with high probability

\begin{align} \label{sleepingsol}
N(j) = n f(j/n) + o(n),
\end{align}
uniformly over $j$, where

\begin{align} \label{sleepinfucntion}
f(x) =\left(1 - \frac{2x}{k}\right)^{\frac{k}{2}},
\end{align}
which is the unique solution to the differential equation
\begin{align} \label{sleepingde}
f'(x) = - \frac{k f(x)}{k - 2x},
\end{align}
with initial condition $f(0) = 1$. Here $\epsilon_0 > 0$ is a fixed constant.

	The result in \cite{wormaldDiff2} is for a fairly general setting, but the bound given for the probability of convergence is not strong enough for us to prove a.s. convergence. Thus, in the following lemma, we specialize the result of Wormald (Theorem 1 in \cite{wormaldDiff2}) to our case so we can obtain a stronger bound on the probability of event (\ref{sleepingsol}).


\begin{lemma} \label{lem:sleepingsolimproved}
Fix a constant $\epsilon_0 > 0$. For any iteration $0 \leq j < (k/2 - \epsilon_0) n$, with probability $1 - o(n^{-3})$, 		
\begin{align*}
|N(j) - n f(j/n)| \leq \delta_1(n),
\end{align*}
uniformly over $j$, where the function $f(x)$ is defined in (\ref{sleepinfucntion}) and $\delta_1(n) = o(n)$.
\end{lemma}

Next we find an approximation for $A(j)$: define $g(x) \triangleq k(1 - f(x)) - 2x$. The following corollary is an immediate corollary of Lemma \ref{lem:sleepingsolimproved}.

\begin{corollary} \label{cor:activesol}
With probability $1 - o(n^{-3})$, for any $0 < j < (k - \epsilon_0) n/2$, we have
\begin{align}
|A(j) - ng(j/n)| \leq \delta_2(n),
\end{align}
uniformly over $j$, for some $\delta_2(n) = o(n)$.
\end{corollary}

We now find an approximation for the number of iterations needed to have $\alpha n$ active or dead nodes. Note that when coupled with the {adoption} process, these are the adopter nodes.
Let $\delta(n) = \max{\{\delta_1(n), \delta_2(n)\}}$ for $\delta_1(n)$ and $\delta_2(n)$, as defined in Lemma \ref{lem:sleepingsolimproved} and Corollary \ref{cor:activesol}.
Let $\mathbf{A} \triangleq (A(j), \ 0 \leq j \leq (k - \epsilon_0)n/2)$ be the vector of the random number of active clones and let $\mathbf{a}$ be a particular realization of this random vector. Also, let $\mathcal{S} \triangleq \{\mathbf{a} : |a(j) - ng(j/n)| \leq \delta(n) \}$. In the next auxiliary lemma, we show the number of iterations needed to have $\alpha n$ active or dead nodes, denoted by $J_{\alpha}$, is close to $f^{-1}(1-\alpha)$.

\begin{lemma} \label{lem:timebound}
Let $j_{\alpha} \triangleq f^{-1}(1-\alpha)$ and $c = |f'\left( \frac{k - \epsilon_0}{2} \right)|$.
If $\mathbf{A} \in \mathcal{S}$, for any $\alpha$
\begin{align*}
|J_{\alpha} - n j_{\alpha}| &\leq \frac{2}{c} \delta(n).
\end{align*}
\end{lemma}

\subsection{Coupling Exploration and Adoption Processes}
\label{subsec:coupling}

We now couple the exploration and {adoption} processes. At any time, the set of active and dead nodes correspond to
the set of adopters. At a random time, one of these nodes ``times-out" and decides to use one of its outgoing edges (or one of its clones, say clone $c$) to contact a neighbor
	at the other end of this outgoing edge. Then, one of the following happens:

	
\begin{itemize}
\item[Case 1:] Clone $c$ has already been paired with clone $c'$.
\begin{itemize}
\item[(a)] Clones $c$ and $c'$ are dead clones.
\item[(b)] The node to which clone $c'$ belongs has already adopted.
\end{itemize}
\item[Case 2:]
Clone $c$ has not been paired yet. It chooses clone $c'$ uniformly at random among all living clones.
\begin{itemize}
\item[(a)] Clone $c$ is an active clone.
\item[(b)] If clone $c'$ is a sleeping clone, the node to which it belongs will adopt; if clone $c'$  is an active clone,  the node to which it belongs has  adopted before this contact.
\end{itemize}
\end{itemize}

This implies only contacts made through active clones will ensure the exploration process proceeds. Further, only these contacts may result in the growth of the adoption set.
As mentioned, we can view the {adoption process} in the following way: each clone (or half edge) makes contact at an independent Poisson process at rate $\beta/k$. Therefore, at any iteration $j$, the time it takes to add one edge (go to iteration (j+1)) is an exponential random variable with rate $\frac{\beta}{k}A(j)$. We denote this random variable by $\tilde{\tau}(j)$. Using the notation defined in Section \ref{sec:modelNmain} and in the previous subsection, it follows that

\begin{align*} 
	\Delta_n(\alpha n, \gamma n) = \sum_{j = J_{\alpha}}^{J_{\gamma}} \tilde{\tau}(j).
\end{align*}

%

\begin{remark}	\label{rem:independence}
		Conditioned on $\left\{\mathbf{A} = a\right\}$, the random variables $\tilde{\tau}(j)$ for $0 < j < (k - \epsilon_0) n/2$ are independent exponentially distributed random variables with mean $\frac{k}{\beta A(j)}$.
\end{remark}

First, we compute $\E{\Delta_n(\alpha n, \gamma n) | \mathbf{A} = a \in \mathcal{S}}$ . By Lemma \ref{lem:timebound}, for $a \in \mathcal{S}$, we have $\underline{\Delta} \leq \Delta_n(\alpha n, \gamma n) \leq \overline{\Delta}$, where

\begin{align*}
\underline{\Delta} \triangleq \sum_{n j_{\alpha} +  \frac{2}{c} \delta(n)}^{n j_{\gamma} -  \frac{2}{c} \delta(n)} \tilde{\tau}(j), \quad \tm{and}  \quad
\overline{\Delta} \triangleq \sum_{n j_{\alpha} -  \frac{2}{c} \delta(n)}^{n j_{\gamma} +  \frac{2}{c} \delta(n)} \tilde{\tau}(j)
\end{align*}

\noindent{Next, we compute}
\begin{align*}
\mathbf{E}[\overline{\Delta} | \mathbf{A} = a \in \mathcal{S}] &= \sum_{n j_{\alpha} - \frac{2}{c} \delta(n)}^{ n j_{\gamma} +  \frac{2}{c} \delta(n)} \frac{k}{\beta A(t)}
 \geq   \frac{k}{\beta}\sum_{n j_{\alpha} - \frac{2}{c} \delta(n)}^{n j_{\gamma} +  \frac{2}{c} \delta(n)} \frac{1}{n g(t/n) + \delta(n)} \\
&=  \frac{k}{\beta} \int_{j_{\alpha}}^{j_{\gamma}} \frac{1}{g(x)} dx    +   o(1)
 = \tilde{\theta}(\gamma) - \tilde{\theta}(\alpha) + o(1),
\end{align*}
where the inequality holds by Corollary \ref{cor:activesol}. In a similar way, we can prove an upper bound for $\mathbf{E}[\overline{\Delta} | \mathbf{A} = a \in \mathcal{S}]$. These bounds prove that

\begin{align} \label{eq:limit:exp}
\mathbf{E}[\overline{\Delta} | \mathbf{A} = a \in \mathcal{S}] \rightarrow \tilde{\theta}(\gamma) - \tilde{\theta}(\alpha).
\end{align}

\noindent{Now, we are in a position to prove Theorem \ref{thm:random}. Let $\mathcal{E}\triangleq \overline{\Delta}- \E{\overline{\Delta}}$. Fix $\epsilon >0$. 	We start with}

\begin{align}\label{eq:bouding}
	\mathbf{P}(|\mathcal{E}| \geq \epsilon) &\leq \sum_{{a} \in \mathcal{S}} \mathbf{P}(|\mathcal{E}| \geq \epsilon | \mathbf{A}  = a) \mathbf{P}(\mathbf{A} = a)
	+ \mathbf{P}(\mathbf{A} \notin \mathcal{S})
\end{align}
 By Corollary \ref{cor:activesol}, the second term is $o(n^{-3})$. To bound the first term, we first notice that for any $\epsilon_0 n \leq j \leq kn/2  - \epsilon_0 n$ and any $p > 0$, conditioned on $a \in \mathcal{S}$, we have:
\begin{align} \label{expbound}
	\mathbf{P}\left(\tilde{\tau}(j) \geq \frac{k p \log n}{\beta \underline{g} n} \Big|\mathbf{A}  = a \right) \leq n^{-p},
\end{align}
where $\underline{g} \triangleq \frac{1}{2} n \min_{x \in [\epsilon_0, k/2 - \epsilon_0]} g(x)$, and we use the following remark:

\begin{remark} \label{activelower}
If $\mathbf{A} \in \mathcal{S}$, then $A(j) \geq n \min_{x \in [\epsilon_0, k/2 - \epsilon_0]} g(x) + o(n) \geq \frac{1}{2} n \min_{x \in [\epsilon_0, k/2 - \epsilon_0]} g(x)$, for all $\epsilon_0 n \leq j \leq (k/2 - \epsilon_0)n$ and large enough $n$.
\end{remark}

Let ${B}$ be the event that $\left\{\tilde{\tau}(j) \leq \frac{kp \log n}{\beta \underline{g} n}\right\}$ for all $\epsilon_0 n \leq j \leq kn/2  - \epsilon_0 n$ and conditioned on $a \in \mathcal{S}$.
By \eqref{expbound} and the union bound, we have $\P{B} \geq 1 - n^{1-p}$, or equivalently, $\P{\overline{B}} \leq n^{1-p}$. Going back to bounding the first term of \eqref{eq:bouding}, for any $a \in \mathcal{S}$ it follows that
\begin{align*}
\mathbf{P}(|\mathcal{E}| \geq \epsilon | \mathbf{A}  = a) \leq \mathbf{P}(|\mathcal{E}| \geq \epsilon | \mathbf{A}  = a, B) + n^{1-p}.
\end{align*}

{\noindent Applying the Hoeffding bounding to the first term, we  get}
\begin{align*}
\mathbf{P}(|\mathcal{E}| \geq \epsilon | \mathbf{A}  = a) \leq
 2\exp\left(-\frac{2n}{\log n} \frac{\epsilon^2 \beta^2 \underline{g}^2}{k^2 p^2} \right) + O(n^{1 - p}).
\end{align*}

{\noindent Using \eqref{eq:limit:exp}, for a  large enough $n$, we have: }
\begin{align*}
\mathbf{P}(|\overline{\Delta} - [\tilde{\theta}(\gamma) - \tilde{\theta}(\alpha)]| \geq 2\epsilon | \mathbf{A}  = a) \leq
 2\exp\left(-\frac{2n}{\log n} \frac{\epsilon^2 \beta^2 \underline{g}^2}{k^2 p^2} \right) + O(n^{1 - p}).
\end{align*}

{\noindent We can prove a similar bound for $|\underline{\Delta} - [\tilde{\theta}(\gamma) - \tilde{\theta}(\alpha)]|$. Since $\underline{\Delta} \leq \Delta_n(\alpha n, \gamma n) \leq \overline{\Delta}$, the same bound holds for $|\Delta_n(\alpha n, \gamma n) - [\tilde{\theta}(\gamma) - \tilde{\theta}(\alpha)]|$}. Let $p = 4$. So far we have proved that

\begin{align}\label{eq:result:config}
\mathbf{P}(|\Delta_n(\alpha n, \gamma n) - [\tilde{\theta}(\gamma) - \tilde{\theta}(\alpha)]| \geq 2 \epsilon)  = O(n^{-3}).
\end{align}

The last step of the proof is to show that when conditioning on the event that the graph obtained as a result of matching random clones is a simple and connected graph, the same order of error probability (i.e., $O(n^{-3})$) holds. Let  $\mathcal{G}$ denote the graph eventually obtained after $nk/2$ iterations. Conditioned on being connected and simple, $\mathcal{G}$ is a $k$-regular graph sampled uniformly among all connected simple $k$-regular graphs. It follows from inequality \eqref{eq:result:config} that

\begin{align}\label{eq:result:simple}
\mathbf{P}\left(\left\{|\Delta_n(\alpha n, \gamma n) - [\tilde{\theta}(\gamma) - \tilde{\theta}(\alpha)]| \geq 2 \epsilon \right\}\cap \left\{\{\mathcal{G} \mbox{ is connected}\} \cap \{ \mathcal{G} \mbox{ is simple}\} \right\}\right)
=  O(n^{-3}).
\end{align}


To prove that conditioned on $\left\{\{\mathcal{G} \mbox{ is connected}\} \cap \{ \mathcal{G} \mbox{ is simple}\} \right\}$, the error probability is $O(n^{-3})$, it suffices to show  $\mathbf{P}(\{\mathcal{G} \mbox{is connected}\} \cap \{ \mathcal{G} \mbox{ is simple}\} ) = O(1)$. This follows from the next two lemmas.

\begin{lemma} \label{lem:probsimple}
$\mathbf{P}(\mathcal{G} \mbox{ is simple})  = \Theta(1)$.
\end{lemma}
The above lemma is a well-known result (see for example \cite{Wormald}, \cite{Bender}). It allows us to carry over properties which hold with high probability from the configuration model to the subset of simple random graphs. The next lemma shows a uniform sample among simple graphs is connected with probability $1 - O(n^{-2})$:

\begin{lemma} \label{lem:connectivity} $\mathbf{P}( \{\mathcal{G} \mbox{ is not connected}\} \cap \{ \mathcal{G} \mbox{ is simple}\})  = O(n^{-2})$.
\end{lemma}


Putting \eqref{eq:result:simple} and Lemmas \ref{lem:probsimple} and \ref{lem:connectivity} together implies that for any graph $G_n$ sampled uniformly at random from the set of all connected $k$-regular graphs with $n$ nodes and a large enough $n$,

\begin{align}\label{eq:result:final}
\mathbf{P}\left(|\Delta_n(\alpha n, \gamma n) - [\tilde{\theta}(\gamma) - \tilde{\theta}(\alpha)]| \geq 2 \epsilon\right) =  O(n^{-3}) < n^{-2}.
\end{align}

Since $\epsilon$ is arbitrary, and $\sum_{n=1}^{\infty} n^{-2}  < \infty$, the theorem (the limit \eqref{eq:timeLimitR}) follows by applying the Borrel-Cantelli lemma.

\section{Discussions}
\label{sec:discus}


The techniques  developed here to analyze the {adoption process} for random $k$-regular graphs can be generalized to the following settings.

{\bf Random graphs with given degree distributions:} We can analyze the {adoption process} defined in Section \ref{sec:modelNmain} for a more general class of random graphs with a given degree distribution. For these random graphs, we couple the {adoption process} and the configuration model. We first analyze the exploration process. Let $D$ be the set of all degrees with a nonzero probability in the degree distribution. Let $N_d(j)$ and $A_d(j)$ to be the number of sleeping nodes of degree $d$ and the number of active clones belonging to a node of degree $d$. We can write equations similar to \eqref{evolutionstart}-\eqref{evolutionend} for the evolution of the vectors $\left(N_d(j), d \in D\right)$ and $\left(A_d(j), d \in D\right)$. Again, we prove an approximation for $1/n \left(N_d(j), d \in D\right)$ using Wormald's result \cite{wormaldDiff2}. We connect the exploration process to the {adoption process} in the same way as in Subsection \ref{subsec:coupling}.

For comparative purposes, see Figure \ref{fig:Simulations}. The figure shows simulation results comparing the process on three different random graphs with mean degree $5$. The theoretical curve derived in this paper is included for reference and matches
very well with the simulation results. Interestingly, when there is more heterogeneity in the degree distribution in the network, the process proceeds more or less at the same rate in the beginning, but  slows down towards  the  end for those with higher disparity.
Arguably,
towards the end of the process, most of the higher degree nodes have already adopted, and the rate of spread is limited
by the remaining lower degree nodes.

\begin{figure}[tbh]
\centering
{\includegraphics[width=16cm,height=7.5cm]{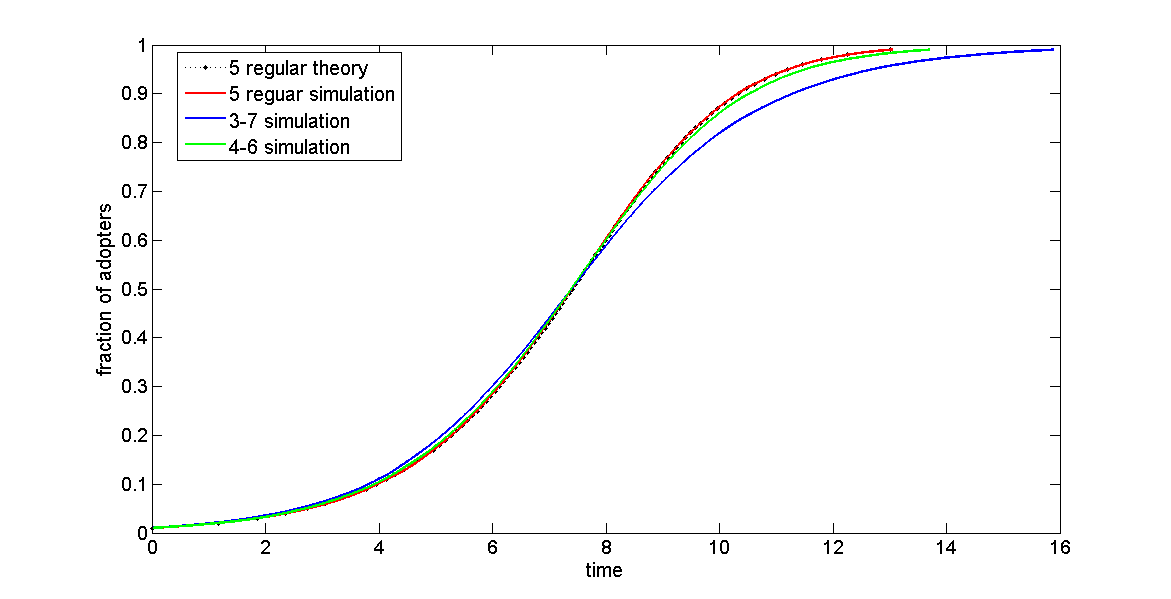}}
\caption{{{Simulated evolution of the fraction of adopters for 5-regular, 4-6 and 3-7 random graphs. Time is normalized such that at time 0 the fraction of adopters is 0.01 for all graphs.}}}
\label{fig:Simulations}
\end{figure}



{
{\bf Adoption process with innovators:} We extend our analysis to include ``innovators'', the  agents who adopt without being contacted by an adopter. In particular, suppose each non-adopter adopts at an independent Poisson process with rate $\beta'$. To include an innovator in the set of adopters, we modify our analysis in the following way. In the exploration process, at iteration $j$, with probability $\frac{\beta' N(j)}{\beta' N(j) + \beta A(j)/k}$, a sleeping node wakes up and all of its $k$ clones become active.  Note this is the probability of a non-adopter becoming an “innovator” before one of the
adopters makes a contact through one of its active half-edges.
If this event does not occur (with probability $1 - \frac{\beta' N(j)}{\beta' N(j) + \beta A(j)/k}$), steps (a) and (b) are performed as described. All of the analysis for the evolution of the exploration process and the coupling exploration and {adoption} process can be modified accordingly.}


{\bf SIR model on random graphs:} As mentioned in the introduction, the {adoption process} studied here is a special case of {adoption} processes used to describe the transmission of communicable diseases by individuals or the spread of computer viruses in a network,  known as the SI (Susceptible-Infected) model. We can use the same ideas to analyze a more general model called SIR that works as follows. The process of infection is the same as that shown in the SI model. Each infected node is removed (they can be interpreted as either dead or immune) at a certain rate. A removed node cannot infect its neighbors. Suppose each infected node is removed at a independent Poisson process with rate $\beta' < \beta$. To extend our framework to this model, in the exploration model, we need to keep track of the nodes that are awake (or, equivalently, are infected) but have been removed as well.

\bibliographystyle{chicago}
\bibliography{references}
\appendix
\section*{APPENDIX}
\setcounter{section}{0}

\section{Missing proofs of Theorem \ref{thm:complete}}
\label{appendixA}

\begin{proof}[Proof of Claim \ref{claim:limit}]
The proof of this claim is mainly algebraic and is given for the sake of completeness.

\begin{align} \label{eq:exp:calculation}
\E{\Delta_n(\gamma n, \alpha n)} & = \sum_{i = \alpha n}^{\gamma n -1} \frac{1}{\lambda_i}
= \frac{n-1}{\beta n}  \sum_{i = \alpha n}^{\gamma n -1} \left( \frac{1}{i} + \frac{1}{n-i}\right) \nonumber\\ & = \frac{n-1}{\beta n} \left[\int_{\alpha}^{\gamma} \frac{1}{x} dx + \int_{\alpha}^{\gamma} \frac{1}{1 -x} dx + \mathcal{E}_n\right] \nonumber \\ &
= \frac{1}{\beta} \left[ \log{\frac{\gamma}{1 - \gamma} - \log{\frac{\alpha}{1 - \alpha}}} \right] + \mathcal{E'}_n,
\end{align}

\noindent{where $\mathcal{E}_n$ and $\mathcal{E'}_n$ represent the error terms. We are left to show that $\mathcal{E'}_n \rightarrow 0$. Note that $\sum_{i = \alpha n}^{\gamma n -1} \frac{1}{i}$ gives an upper-bound for the integral $\int_{\alpha}^{\gamma} \frac{1}{x} dx$. Further,$\sum_{i = \alpha n}^{\gamma n -1} \frac{1}{i+1}$ gives an lower-bound for it. Now the difference between these two sums is $\Theta(1/n)$. Similarly we can show the error contribution of the second sum is also order $\Theta(1/n)$ which completes the proof.}
\end{proof}

\begin{proof}[Proof of Claim \ref{claim:concentration}]
In the following proof, to simplify the notation, we present $\Delta_n(\gamma n, \alpha n)$  by $\Delta_n$. By the Chernoff's bound, for any  $s \geq 0$:
\begin{align*}
\P{{\Delta_n- \E{\Delta_n} \geq \epsilon}}\leq e^{ - s \epsilon} \E{e^{s \left(\Delta_n - \E{\Delta_n}\right)}} = e^{-s \left(\epsilon + \E{\Delta_n} \right)} \E{e^{s \Delta_n}}.
\end{align*}
 Note that $\Delta_n$ is sum of independent exponential random variables, thus we can compute $\E{e^{s \Delta_n }}$ as follows:
\begin{align*}
\E{e^{s \Delta_n }} =  \prod_{i = \alpha n} ^{\gamma n - 1} \left(1 + \frac{s}{\lambda_i - s}\right).
\end{align*}
For any real number $z$, we have $(1 + z) \leq e^{z}$. Therefore,

\begin{align*}
\E{e^{s \Delta_n }} \leq e^{\sum_{i = \alpha n} ^{\gamma n  -1} \frac{s}{\lambda_i - s}}
\end{align*}

Let $\lambda^*$ be $\min\left\{\lambda_{\alpha n}, \lambda_{\gamma n -1}\right\}$, and define $s^* = \frac{\frac{\epsilon}{2 \E{\Delta_n}}}{1 + \frac{\epsilon}{2 \E{\Delta_n}}} \lambda^*$.
Because for any $\alpha n \leq i \leq \gamma n -1$, $\lambda_i \geq \lambda^*$, it follows that:

\begin{align*}
\frac{1}{\lambda_i - s^*} \leq \frac{1}{\lambda_i} \left( 1 + \frac{\epsilon}{2 \E{\Delta_n}}\right).
\end{align*}
Setting $s$ to be $s^*$ in the Chernoff's bound, we have:
\begin{align*}
\P{\Delta_n - \E{\Delta_n} \geq \epsilon } \leq e^{  - \frac{ s^* \epsilon}{2}}
\end{align*}
Similarly, we can show that:
\begin{align*}
\P{- (\Delta_n - \E{\Delta_n}) \geq \epsilon } \leq e^{  - \frac{ s^* \epsilon}{2}} .
\end{align*}

\noindent{Define}

\begin{align}
\label{delta}
\delta = \left(\frac{\frac{\epsilon^2}{2 \E{\Delta_n}}}{1 + \frac{\epsilon}{2 \E{\Delta_n}}} \right)\lambda^*/n .
\end{align}

\noindent{Note that $\lambda^*/n$ is a constant bounded away from zero. More precisely,}

\begin{align*}
\lambda^*/n = \min\left\{ \alpha \frac{n - \alpha n} {n-1},  \frac{(\gamma n  -1) (n - \gamma n +1)} {n (n-1)}\right\}.
\end{align*}

\noindent{Thus the inequality \eqref{eq:claim:hoeffding} holds with the above $\delta$.}
\end{proof}
\section{Proof of Theorem \ref{thm:sublin}}
\label{sec:sublin}
\begin{proof}[Proof of part (a)]
To prove this part, we show that $\frac{\E {T_n(C \log n)}}{\log \log n} \rightarrow \frac{1}{\beta} $ and $Var(T_n(C \log n))$ is bounded. Thus the statement follows by applying  Chebyshev's inequality. Using \eqref{eq:rate:complete}, we have:
\begin{align*}
	\E {T_n(C \log n)} = \frac{1}{\beta}  \sum_{i=1}^{C \log n} \frac{n-1}{i(n-i)}  = \frac{1}{\beta} \log \log n + o(\log \log n).
\end{align*}

The steps of the above calculation are similar to the ones presented in \eqref{eq:exp:calculation}, thus we remove the details. Further, note that $T_n(C \log n)$ is the sum of independent exponential random variables, thus:

\begin{align*}
	Var(T_n(C \log n)) = \sum_{i=1}^{C \log n} \frac{1}{\lambda_i^2} = \frac{1}{\beta^2}  \sum_{i=1}^{C \log n}  \left(\frac{n-1}{i(n-i)}  \right)^2 \leq  \frac{1}{\beta^2}  \sum_{i=1}^{\infty}  \left( \frac{n-1}{i(n-i)} \right)^2 \leq M,
\end{align*}
for some constant $M$. This completes the proof.
\end{proof}

\begin{proof}[Proof of part (b)]

Similar to part (a), we show that $\frac{\E {T_n(C \log n)}}{\log \log n} \rightarrow \frac{k}{\beta(k-2)} $ and $Var(T_n(C \log n))$ is bounded. Following the proof sketch, we use the known fact that, w.h.p., a random $k$-regular graph is locally tree like (for instance see \cite{dembo}). For the sake of completeness, in the following lemma, we state this using the terminology for the exploration process described in Section \ref{sec:random}.

\begin{lemma} \label{lem:tree}
For $j = O(\log n)$, w.h.p., the component of the configuration model revealed until $j$ iterations is a tree. \footnote{Note that the result in the literature is stronger than this statement, and it asserts that even after forming the whole graph, the subgraph formed by these nodes is a tree w.h.p.}
\end{lemma}

\begin{proof}[Proof of Lemma \ref{lem:tree}]
	For iteration $j \leq C \log n$ in the exploration process, the number of active clones satisfies $A(j) \leq k C \log n$. Thus the probability that there are no cycles in the multigraph formed so far by the configuration model is at least
	$\left(1 - \frac{k C \log n}{k(n - C \log n)} \right)^{C \log n} = 1 -  o(1)$. This implies that w.h.p., the mutli-graph formed so far, is a tree.
\end{proof}
Lemma \ref{lem:tree} implies that in this early adoption regime, at each step of the exploration process, one sleeping node is awakened and using our coupling with the adoption process, this means one extra node
adopts $Z$. Also, at any time step $j \leq C \log n$, the number of active clones is given by $A(j) = k + j(k-2)$. This gives
\begin{align*}
	\E {T_n(C \log n)} = \frac{k}{\beta}  \sum_{i=1}^{C \log n} \frac{1}{k + i (k - 2)}  = \frac{1}{\beta} \frac{k}{k-2} \log \log n + o(\log \log n).
\end{align*}
Similar to part (a), the variance of $T_n(C \log n)$ can be calculated as:
\begin{align*}
	Var(T_n(C \log n)) = \frac{k^2}{\beta^2}  \sum_{i=1}^{C \log n}  \left( \frac{1}{k + i (k - 2)} \right)^2 \leq  \frac{k^2}{\beta^2}  \sum_{i=1}^{\infty}  \left( \frac{1}{k + i (k - 2)} \right)^2 \leq M,
\end{align*}
for some constant $M$. Part (b) of Theorem \ref{thm:sublin} then follows by using the  Chebyshev's inequality.
\end{proof}

\section{Missing proofs of Section \ref{sec:random}}
\label{appendixB}

\begin{proof}[Proof of Lemma \ref{lem:sleepingsolimproved}]
The proof is a simple modification of Wormald's original proof, where the modification is mainly in how we choose the various asymptotic functions involved in the proof.
Let $F(x,y) = \frac{ky}{k - 2x}$ defined for $0 \leq x \leq k/2 - \epsilon_0$ and $0 \leq y \leq 1$. Then
\begin{align*}
	\left|\frac{\partial F}{\partial x}\right| + \left|\frac{\partial F}{\partial y}\right| =  \frac{2ky}{(k - 2x)^2} + \frac{k}{k - 2x}  \leq \frac{k}{2 \epsilon_0^2} + \frac{k}{ {2}\epsilon_0} \triangleq \Phi.
\end{align*}
So, $F$ is Lipschitz with Lipschitz constant $\Phi$. Also note that for any $0<j<n$, we have $F(j/n, N(j)/n) = f'(j/n)$ (defined in \eqref{sleepingde}). Let  $\lambda  = \frac{n}{(\log n)^2} $ and $\delta = \frac{1}{\log n}$. We first show that:

\begin{claim}
\label{cl:supMartingale}
\begin{align} \label{hoeffding2}
	\mathbf{P}(|N(j+\lambda) - N(j) - \lambda F(j/n, N(j)/n)| \geq  4 \delta \lambda) \leq 2 e^{-\frac{\delta^2 \lambda}{2}}.
\end{align}
\end{claim}

\begin{proof}[Proof of Claim \ref{cl:supMartingale}]
For each $0 \leq l \leq \lambda$, we have
\begin{align*}
\mathbf{E}[ N(j+l+1) - N(j + l) | H(j+l)] = F((j+l)/n, N((j+l)/n)) \leq F(j/n, N(j)/n) + \Phi \lambda/n.
\end{align*}
Thus there exists a function $\Delta(n)  = \Phi \lambda/n$ such that conditional on $H(j)$, the sequence
\begin{align*}
			Z(l) = N(j+l) - N(j) - lF(j/n, N(j)/n) - l \Delta(n)
\end{align*}
for $0 \leq l \leq \lambda$ is a supermartingale w.r.t. the sigma fields generated by $H(j), H(j+1), \ldots, H(j+\lambda)$.
Further $Z(l)$, $0 \leq l \leq \lambda$ has bounded increment: $|Z(l+1) - Z(l)| \leq |N(j+l+1) - N(j+l)| + |F(j/n, N(j)/n)| + |\Delta(n)| \leq 3$.
Thus using concentration results for supermartingales from \cite{wormaldDiff} (Lemma 4.2), we have
\begin{align} \label{hoeffding1}
	\mathbf{P}(N(j+\lambda) - N(j) - \lambda F(j/n, N(j)/n) \geq \lambda \Delta(n) + 3 \delta \lambda) \leq e^{-\frac{\delta^2 \lambda}{2}}.
\end{align}
The following two observations complete the proof: (i) $\lambda \Delta(n)\leq \delta \lambda$ (ii) A similar inequality can be obtained for the lower tail of $Z(\lambda)$.

\end{proof}

Next, for $i = 0,1, ..., (k/2 - \epsilon_0)n/\lambda$, by induction we prove that:

\begin{claim}
\label{cl:inductiion}
$\mathbf{P}(|N(i \lambda) - f(i \lambda /n)| \geq \epsilon_i) \leq i e^{-\frac{\delta^2 \lambda}{2}}$, where $\epsilon_i = 5 \delta \lambda \left[ \frac{(1 + \Delta(n))^i - 1}{\Delta(n)} \right]$.
\end{claim}

\begin{proof}[Proof of Claim \ref{cl:inductiion}]
The base case follows from Claim \ref{cl:supMartingale} (inequality (\ref{hoeffding2})), and the fact that $|f(\lambda) - \lambda f'(0)| \leq \delta \lambda$.
Assuming that the claim holds for $1, 2, \ldots, i$, we now prove it for case $i+1$.
We have
\begin{align*}
|N((i+1)\lambda) - n f((i+1)\lambda/n)| \leq |N(i \lambda) - n f(i \lambda/n)| + |N((i+1)\lambda) - N(i\lambda) +  n f(i\lambda/n) -  n f((i+1)\lambda/n)|.
\end{align*}
By induction hypothesis, with probability $1 - i e^{-\frac{\delta^2 \lambda}{2}}$, we have $|N(i \lambda) - n f(i \lambda/n)|  \leq \epsilon_i$. Further, we write:
\begin{align*}
|N((i+1)\lambda) - N(i\lambda) +  n f(i\lambda/n) -  n f((i+1)\lambda/n)| & \leq |N((i+1) \lambda) - N(i \lambda) - \lambda F(i \lambda/n, N(i \lambda)/n)| \\
&+ |  \lambda F(i \lambda/n, N(i \lambda)/n) + nf(i\lambda/n) -   nf((i+1)\lambda/n)|.
\end{align*}
From (\ref{hoeffding2}), we have with probability $1- e^{-\frac{\delta^2 \lambda}{2}}$, $|N((i+1) \lambda) - N(i \lambda) - \lambda F(i \lambda/n, N(i \lambda)/n)| \leq 4 \delta \lambda$. Also,
\begin{align*}
	nf(i\lambda/n) -   nf((i+1) \lambda/n) &=  - \lambda f'(i\lambda/n) + O(\lambda^2/n) \\
	&= - \lambda F(i\lambda/n, f(i\lambda/n)) + O(\lambda^2/n).
\end{align*}	
Hence,
\begin{align*}
	| \lambda F(i \lambda/n, N(i \lambda)/n) + nf(i\lambda/n) -   nf((i+1)\lambda/n)| &\leq | \lambda F(i \lambda/n, N(i \lambda)/n)  -   \lambda F(i\lambda/n, f(i\lambda/n)) | + O(\lambda^2/n)\\
	& \leq \Delta(n) \epsilon_i + O(\lambda^2/n).
\end{align*}	
where in the inequality holds because $F(x,y)$ is Lipschitz. Note that $\lambda^2/n = o(\delta \lambda)$. Putting all these together,
we have with probability at least $1 - (i+1) e^{-\frac{\delta^2 \lambda}{2}}$,
\begin{align*}
|N((i+1) \lambda) - n f((i+1) \lambda/n)| \leq (1 + \Delta(n)) \epsilon_i + 5 \delta \lambda = \epsilon_{i+1}.
\end{align*}
This complete the induction and the proof.
\end{proof}

To complete the proof of the Lemma, first note that $i = O(n/\lambda)$, so $\epsilon_i = O(n \delta) = o(n)$.
Now for any general $j$, find $i^*$ such that $i^* \lambda$ is the nearest integer to $j$ among all $\lambda i$'s, $i = 0,1,\dots, (k/2 - \epsilon_0)n/w$. We have:

\begin{align*}
|N(j) - n f(j/n)| &\leq |N(j) - N(i^* \lambda) | + |N(i^* \lambda) - n f(i^* \lambda/n)| + n |f(i^* \lambda/n) - f(j/n)| \\
			&= |N(i^* \lambda) - n f(i^* \lambda/n)| + O(\lambda).
\end{align*}
Noting that $\frac{n}{\lambda} e^{-\frac{\delta^2 \lambda}{2}} = o(n^{-3})$ completes the proof.
\end{proof}




\begin{proof}[Proof of Lemma \ref{lem:timebound}]
		First note that for any $0 \leq x \leq \frac{k - \epsilon_0}{2}$, we have $f'(x) \leq - |f'\left( \frac{k - \epsilon_0}{2} \right)| = -c$. Using the mean value theorem we have: $nf(j_{\alpha}+ \frac{2}{c}|\delta(n)|/n) \leq nf(j_{\alpha}) + 2 \delta(n)$. From Lemma \ref{lem:sleepingsolimproved}, we have $N(n j_{\alpha} + \frac{2}{c}|\delta(n)|)  \leq nf(j_{\alpha} + \frac{2}{c}|\delta(n)|/n) + \delta(n) \leq nf(t_{\alpha}) - |\delta(n)| = n (1 - \alpha) -  |\delta(n)|$.
	    By definition, $N(T_{\alpha}) = n(1 - \alpha)$, and $N(\cdot)$ is decreasing in the number if iterations. Thus,  $J_{\alpha} \leq n j_{\alpha} + \frac{2}{c} |\delta(n)|$. Using a similar argument, we can prove $J_{\alpha} \geq n j_{\alpha} - \frac{2}{c} |\delta(n)|$ and the proof is complete.

\end{proof}

\begin{proof}[Proof of Lemma \ref{lem:connectivity}]
Let $E$ be the event that the configuration model is not connected and
let $E_s$ be the event that there exists a subgraph $\mathcal{H}$ of size $s$ which is isolated in the configuration model.
For any positive even integer $m$, the number of possible pairings of $m$ clones is
$\frac{{m!}}{(m/2)! 2^{(m/2)}}$.
Using this fact we can bound the probability of $E_s$ by simply counting the number of
 pairings as follows.
\begin{align*}	
\mathbf{P}(E_s) \leq  { n \choose s} \frac{ {nk/2 \choose sk/2} }{{nk \choose sk}}.
\end{align*}
The above bound is asymptotically largest when $s$ is a constant with respect to $n$. See \cite{Wormald} for similar proof techniques. In this case for any $\epsilon_1 > 0$, we can obtain:
\begin{align} \label{simple}
\mathbf{P}(E)  = O(n^{-s(k/2 - 1) + \epsilon_1}).
\end{align}
Since $k \geq 3$, when $\mathcal{G}$ is simple, it cannot have an isolated component of size less than $4$ nodes. The worst case in (\ref{simple}) is when $k = 3$ and $s = 4$ which gives the required bound.
\endproof
\end{proof}

\end{document}